\documentclass[10pt,onecolumn,twoside]{article}

\usepackage{epsfig}
\usepackage{titlesec}
\usepackage{url} 
\usepackage{booktabs} 
\usepackage{amsfonts} 
\usepackage{amssymb} 
\usepackage{nicefrac} 
\usepackage{microtype} 
\usepackage{mathrsfs}
\usepackage{bm} 
\usepackage{cite} 
\usepackage{comment}
\usepackage{graphicx} 
\usepackage{mathtools} 
\usepackage{algpseudocode}
\usepackage{algorithm}
\usepackage{amsthm} 
\usepackage{enumitem} 
\usepackage{multirow} 
\usepackage{tabularx} 
\usepackage{hhline}
\usepackage{arydshln} 
\usepackage{enumitem} 
\usepackage{siunitx}
\usepackage[font=small,skip=5pt]{caption}
\usepackage{parskip}
\usepackage{comment}
\usepackage[dvipsnames]{xcolor}
\usepackage{fullpage}
\usepackage[hidelinks]{hyperref}

\algtext*{EndIf}
\algtext*{EndFor}
\algtext*{EndWhile}

\newcolumntype{L}[1]{>{\raggedright\arraybackslash}p{#1}}
\newcolumntype{C}[1]{>{\centering\arraybackslash}p{#1}}
\newcolumntype{R}[1]{>{\raggedleft\arraybackslash}p{#1}}

\theoremstyle{plain} 
\newtheorem{proposition}{Proposition}
\newtheorem{definition}{Definition}
\newtheorem{theorem}{Theorem}

\newtheorem{assumption}{Assumption}

\def\defn{\,\coloneqq\,}
\def\argmin{\mathop{\mathsf{arg\,min}}} 

\def\lim{\mathop{\mathsf{lim}}} 
\def\min{\mathop{\mathsf{min}}} 
\def\max{\mathop{\mathsf{max}}}

\def\zer{\mathsf{zer}}

\def\DnCNNast{{\text{DnCNN}^\ast}}
\def\uin{u_\text{in}}
\def\usc{u_\text{sc}}
\def\del{\, \mathrm{d}} 

\def\ebm{{\bm{e}}}
\def\hbm{{\bm{h}}}

\def\xbm{{\bm{x}}}

\def\ybm{{\bm{y}}}

\def\zerobm{\bm{0}}
\def\Abm{{\bm{A}}}
\def\Hbm{{\bm{H}}}

\def\Fbm{{\bm{F}}}

\def\Ibm{{\bm{I}}}

\def\ybf{{\mathbf{y}}}
\def\rbmp{{\bm{r}^\prime}}

\def\xbmast{{\bm{x}^\ast}}

\def\xbmhat{{\widehat{\bm{x}}}}

\def\Gsfhat{{\widehat{\Gsf}}}

\def\nablahat{{\widehat{\nabla}}}

\def\rbmp{{\bm{r}^\prime}}
\def\rbfp{{\mathbf{r}^\prime}}

\def\rbf{\mathbf{r}}

\def\ybf{\mathbf{y}}

\def\ebf{\mathbf{e}}

\def\Ibf{\mathbf{I}}

\def\Tsf{{\mathsf{T}}}

\def\Dsf{{\mathsf{D}}}

\def\Hsf{{\mathsf{H}}}
\def\Nsf{{\mathsf{N}}}
\def\Gsf{{\mathsf{G}}}
\def\Isf{{\mathsf{I}}}
\def\Psf{{\mathsf{P}}}
\def\Rsf{{\mathsf{R}}}

\def\Hsf{{\mathsf{H}}}

\def\C{\mathbb{C}}
\def\R{\mathbb{R}}
\def\E{\mathbb{E}}

\def\Lcal{{\mathcal{L}}}
\def\De{{\Delta \epsilon}}
\def\DeRe{{\Delta \epsilon_\mathrm{Re}}}
\def\DeIm{{\Delta \epsilon_\mathrm{Im}}}
\def\tDeRe{{\widetilde{\Delta \epsilon}_\mathrm{Re}}}
\def\tDeIm{{\widetilde{\Delta \epsilon}_\mathrm{Im}}}
\def\HRe{{H_{\mathrm{Re}}}}
\def\HIm{{H_{\mathrm{Im}}}}

\def\De{{\Delta \epsilon}}
\def\DeRe{{\Delta \epsilon_\text{Re}}}
\def\DeIm{{\Delta \epsilon_\text{Im}}}
\def\tDeRe{{\widetilde{\Delta \epsilon}_\text{Re}}}
\def\tDeIm{{\widetilde{\Delta \epsilon}_\text{Im}}}
\def\HRe{{H_{\text{Re}}}}
\def\HIm{{H_{\text{Im}}}}

\usepackage{tikz}
\newcommand{\pgftextcircled}[1]{
    \setbox0=\hbox{#1}%
    \dimen0\wd0%
    \divide\dimen0 by 2%
    \begin{tikzpicture}[baseline=(a.base)]%
        \useasboundingbox (-\the\dimen0,0pt) rectangle (\the\dimen0,1pt);
        \node[circle,draw,outer sep=0pt,inner sep=0.1ex] (a) {#1};
    \end{tikzpicture}
}
\let\textcircled=\pgftextcircled

\def\ctwo{\textcircled{2}}

\def\cfive{\textcircled{5}}
\def\csix{\textcircled{6}}

\def\proposed{SIMBA}

\definecolor{pink}{HTML}{000000}
\definecolor{lightgreen}{RGB}{229,251,229}
\definecolor{red}{rgb}{0,0,0}

\title{SIMBA: Scalable Inversion in Optical Tomography\\using Deep Denoising Priors}

\author{Zihui~Wu%
\thanks{Department of Computer Science \& Engineering, Washington University in St.~Louis, St.~Louis, MO 63130.}
\hspace{0.05em},
Yu~Sun$^{\ast}$,
Alex~Matlock%
\thanks{Department of Electrical and Computer Engineering, Boston University, Boston, MA 02215, USA.}
\hspace{0.05em},
Jiaming~Liu%
\thanks{Department of Electrical \& Systems Engineering, Washington University in St.~Louis, St.~Louis, MO 63130.}
\hspace{0.05em},\\
Lei~Tian$^{\dagger}$%
, and Ulugbek~S.~Kamilov$^{\ast, \ddagger}$}

\begin{document}
\date{}
\maketitle

\begin{abstract}
Two features desired in a three-dimensional (3D) optical tomographic image reconstruction algorithm are the ability to reduce imaging artifacts and to do fast processing of large data volumes. Traditional iterative inversion algorithms are impractical in this context due to their heavy computational and memory requirements. We propose and experimentally validate a novel \emph{scalable iterative minibatch algorithm (SIMBA)} for fast and high-quality optical tomographic imaging. SIMBA enables high-quality imaging by combining two complementary information sources: the physics of the imaging system characterized by its forward model and the imaging prior characterized by a denoising deep neural net. SIMBA easily scales to very large 3D tomographic datasets by processing only a small subset of measurements at each iteration. We establish the theoretical fixed-point convergence of \proposed~under nonexpansive denoisers for convex data-fidelity terms. We validate \proposed~on both simulated and experimentally collected intensity diffraction tomography (IDT) datasets. Our results show that SIMBA can significantly reduce the computational burden of 3D image formation without sacrificing the imaging quality.
\end{abstract}

\section{Introduction}

Optical tomographic imaging seeks to recover the three-dimensional (3D) distribution of the refractive index of an object from its light measurements. 
In a standard setup (see Figure~\ref{fig:Schema} for an example), the sample is illuminated multiple times from different angles and the scattered light-field is recorded with a camera. In the interferometry-based microscopy, one measures both the amplitude and the phase of the scattered field~\cite{Choi.etal2007,Kamilov.etal2015,Kim.etal2014}, while in the intensity-only setups one measures only the amplitude of the light-field~\cite{Gbur.etal02, Tian.etal14, Tian.Waller2015}. A tomographic reconstruction algorithm is then used to computationally reconstruct the 3D distribution of the sample's refractive index. The quantitative characterization of the refractive index is important in biomedical imaging since it allows to visualize the internal structure of a tissue, as well as characterize physical changes within biological samples.

\begin{figure}[t]
\centering\includegraphics[width=8.5cm]{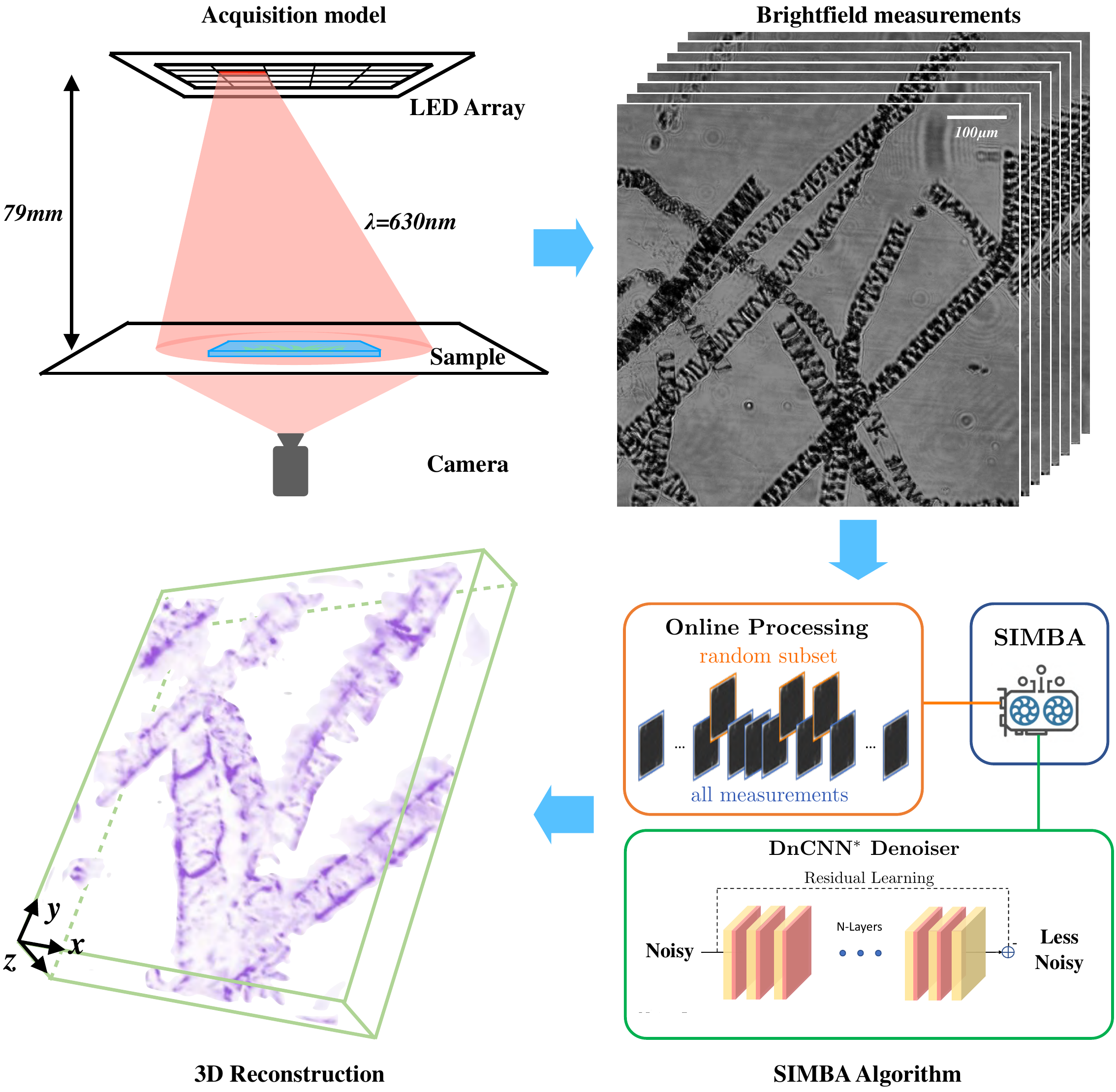}
\caption{The conceptual illustration of the proposed inversion algorithm for optical tomographic imaging. The brightfield measurements of the scattered light-field are collected with a standard computational microscope platform. An \emph{online} reconstruction algorithm, SIMBA, facilitated by a \emph{convolutional neural network (CNN)} denoiser is then used to form a 3D phase image. Unlike traditional batch algorithms, SIMBA processes a subset of measurements at a time, making it scalable for processing very large tomographic datasets.}
\label{fig:Schema}
\end{figure}

The reconstruction of the refractive index is often formulated as an inverse problem. In this context, the \emph{forward model} characterizes the physics of data-acquisition and can be used to ensure the consistency of the final estimate with respect to the measurements. However, the need for processing large-scale tomographic data limits the utility of traditional iterative methods in 3D optical tomography. Traditional \emph{batch} algorithms process the whole tomographic dataset at every iteration. On the other hand, \emph{online} algorithms can effectively scale to large datasets by processing only a small subset of data per iteration.



Using imaging priors is a standard strategy for mitigating the ill-posed nature of many tomographic imaging problems. Popular imaging priors include Tikhonov~\cite{Ribes.Schmitt2008} and total variation (TV)~\cite{Kamilov.etal2016} regularizers. Recently, a new class of methods, called \emph{plug-and-play priors (PnP)} \cite{Venkatakrishnan.etal2013}, have popularized the idea of using general image denoisers as imaging priors within iterative inversion. By leveraging advanced image denoisers, such as BM3D~\cite{Danielyan2013} and DnCNN~\cite{Zhang.etal2017}, PnP methods have achieved the state-of-the-art performance in various imaging applications~\cite{Sreehari.etal2016, Chan.etal2016, Teodoro.etal2019, Brifman.etal2016, Teodoro.etal2016, Zhang.etal2017a, Meinhardt.etal2017, Kamilov.etal2017, Sun.etal2019a, Sun.etal2019b}. An alternative framework for using image denoisers is the \emph{regularization by denoising (RED)}~\cite{Romano.etal2017}, where the denoiser is used to formulate an explicit regularizer that has a simple gradient. The work~\cite{Reehorst.Schniter2019} has clarified the existence of RED regularizers for certain class of denoisers, and the excellent performance of the framework has been demonstrated in phase retrieval~\cite{Metzler.etal2018} and image super-resolution~\cite{Mataev.etal2019} using DnCNN and the deep image prior, respectively. In short, using advanced denoisers has proven to be effective for improving the reconstruction quality in various imaging contexts. Note that the concept of regularization described in this paper is distinct from the concept of regularization for training deep neural nets by injecting noise~\cite{Seghouane.etal2004}.

In this paper, we present a new \emph{scalable iterative minibatch algorithm (SIMBA)} for the regularized inversion in optical tomography. \proposed~is an online extension of the traditional RED framework. It can thus leverage powerful convolutional neural network (CNN) denoisers as imaging priors, while also taking advantage of the physical information available through the forward model. However, unlike traditional RED algorithms, SIMBA is scalable to datasets that are too large for batch processing since it only uses a subset of measurements at a time. We prove that \proposed~converges in expectation to the same set of fixed points as its batch counterparts under a set of transparent assumptions. Thus, SIMBA benefits from the excellent imaging quality offered by RED, but does so in a computationally tractable way for optical tomographic imaging.

We validate \proposed~in the context of intensity-only microscopy called \emph{intensity diffraction tomography (IDT)}. IDT microscopes are relatively cheap and easy to implement since they do not collect the phase of the light. We adopt the IDT forward model in~\cite{Ling.etal18} that establishes a linear relationship between the desired object and the intensity measurements by neglecting the terms corresponding to higher order light scattering. We show that \proposed~can efficiently reconstruct a high-resolution ($1024\times1024\times25$ pixels) IDT image while also offering improvements in the 3D sectioning capability. The preliminary version of this work was presented in~\cite{Wu.etal2019}. The current paper significantly extends~\cite{Wu.etal2019} by including the IDT model, providing additional simulations, and validating the method on an experimentally collected 3D IDT dataset.

This paper is organized as follows. In Section~\ref{Sec:Background}, we introduce the IDT forward model and the RED framework. In Section~\ref{Sec:Proposed}, we present the algorithmic details of \proposed. In Section~\ref{Sec:Convergence}, we analyze the fixed-point convergence under a set of assumptions. In Section~\ref{Sec:Experiments}, we provide simulations and experiments that illustrates the efficiency and effectiveness of \proposed. Section~\ref{Sec:Conclusion} concludes the paper.

\section{Background}
\label{Sec:Background}
In this section, we provide the background on IDT and image-denoising priors. We start by describing the IDT forward model, then formulate the corresponding inverse problem, and finally introduce the RED framework as a strategy to leverage image denoisers as priors.

\subsection{Linearized IDT}

Consider a 3D object with the permittivity distribution $\epsilon(\rbf)$ in a bounded sample domain $\Omega \subset \R^3$, immersed into the background medium of permittivity $\epsilon_b$. We use $\De = \DeRe + i\DeIm = \epsilon - \epsilon_b$ to denote the permittivity contrast between the object and the background medium. The real part $\DeRe$ corresponds to the phase effect, and the imaginary part $\DeIm$ accounts for the absorption. The object is illuminated by an angled incident light field $\uin(\rbf)$. The incident field $\uin$ is assumed to be known inside $\Omega$ as well as at the camera domain $\Gamma \subset \R^2$. The total light-field $u(\rbf)$ is measured only through its intensity at the camera. Here, $\rbf = (x,y,z)$ denotes the 3D spatial coordinates. Under the first Born approximation~\cite{Wolf1969}, the light-sample interaction is described by the following equation
\begin{equation}
\label{Eq:ImageField}
u(\rbf) = \uin(\rbf) + \int_{\Omega} g(\rbf - \rbmp) \, v(\rbfp) \, \uin(\rbfp) \del \rbfp,\quad\rbf \in \Omega
\end{equation}
where $u(\rbf) = \uin(\rbf) + \usc(\rbf)$ is the total light field, ${v(\rbf) = \frac{1}{4 \pi} k^2 \De}$ is the scattering potential, $k=2\pi/\lambda$ is wave number in free space, and $\lambda$ is the wavelength of the illumination. In the 3D space, the Green's function at the camera plane $\Gamma$ is given by
$$g(\rbf) = \frac{\ebf^{ik_b\|\rbf\|_2}}{\|\rbf\|_2},$$
where $k_b = \sqrt{\epsilon_b}k$ is the wavenumber of the background medium, and $\|\cdot\|_{2}$ denotes the $\ell_2$-norm. For a single illumination, the intensity of the light field after propagating through the sample is given by
\begin{equation}
\label{Eq:NonlinearModel}
\Ibf= |u(\rbf) \ast p|^2,
\end{equation}
where $p$ is the point spread function of the microscope, and the operator $\ast$ denotes the 2D convolution. Eq.~\eqref{Eq:NonlinearModel} can be expanded into the summation of four components
\begin{equation}
\label{Eq:Expansion}
\Ibf= \Ibf^{ii}+\Ibf^{ss}+\Ibf^{is}+\Ibf^{si},
\end{equation}
where $\Ibf^{ii}$ is the constant background intensity, $\Ibf^{ss}$ is the squared modulus of the scattered field, and $\Ibf^{is}=(\Ibf^{si})^{\ast}$ are the cross terms that relate the unscattered and scattered field. Here, $(\cdot)^{\ast}$ denotes the complex conjugate. Due to the first Born approximation, $\Ibf^{ss}$ can be assumed to be small and thus neglected. By modeling the 3D object as a series of slices along the axial dimension $z$, one can represent the spectrum of the total scattered field as the summation of the sub-scattered fields produced by each slice~\cite{Ling.etal18}
\begin{equation}
\label{Eq:LinearizedModel}
\widetilde{\Ibf} = \widetilde{\Ibf}^{ii} \; + \int\left[\HRe(z) \tDeRe(z)+\HIm(z) \tDeIm(z)\right] \del z
\end{equation}
where $\widetilde{\,\cdot\,}$ denotes 2D Fourier transform, and $\widetilde{\Ibf}^{ii}$ is the background intensity spectrum measured at $\Gamma$. In (\ref{Eq:LinearizedModel})$, \HRe$ and $\HIm$ are the angle-dependent phase and absorption transfer functions (TF) for each sample slice at depth $z$, respectively. These TFs linearly map the Fourier transform of the permittivity contrast to the intensity spectrum of the scattered field. We refer the reader to~\cite{Ling.etal18} for the full details of the TF for IDT.

By discretizing \eqref{Eq:LinearizedModel} and explicitly including the Fourier transform into the equation, we obtain the following linear model in the spatial domain for the $i^{\text{th}}$ illumination
\begin{equation}
\label{Eq:discretizedSingle}
\Ibm_{i} = \Ibm^{ii}_i + \Re\Bigg\{\sum_{j=0}^{J}\Abm_{ij}\xbm_{j}\Bigg\},\;\,\,\text{with}\;\,\,\Abm_{ij}=\Fbm^{\textsf{\tiny H}}\Hbm_{ij}\Fbm,
\end{equation}
where $j=0,\dots, J$ discretely indexes the axial direction $z$, $\xbm_j \in \C^N$ is the discretized complex permittivity contrast of the $j^{\text{th}}$ slice, $\Ibm_i$ is the measured intensity of the total field, $\Ibm^{ii}_i$ is the discretized intensity of the background, and $\Hbm_{ij}$ is the discretized TF accounting for both phase and absorption at $z_j$. We use $\Fbm$ and $\Fbm^{\textsf{\tiny H}}$ to denote the 2D discrete Fourier transform and its inverse, respectively. By re-arranging the terms, we can obtain the following linear forward model
\begin{equation}
\label{Eq:discretizedModel}
\ybm_i = \Abm_i\xbm + \ebm, \quad\text{with}\quad\Abm_i^\textsf{\tiny H} = \begin{bmatrix}\Abm_{i0}^{\textsf{\tiny H}}\\ \vdots \\ \Abm_{iJ}^\textsf{\tiny H}\end{bmatrix},\; \xbm = \begin{bmatrix}\xbm_0\\ \vdots \\ \xbm_J\end{bmatrix}
\end{equation}
where the operator $(\cdot)^{\Hsf}$ denotes the conjugate transpose, ${\ybm_i \defn \Ibm_i -\Ibm_i^{ii} \in \R^N}$ is the measured intensity with the removal of the background intensity for the $i^{\text{th}}$ illumination, and $\ebm\in\R^N$ is the error term. Note that, as was discussed in~\cite{Ling.etal18}, the IDT forward model does not contain any information on the DC component of the phase.

\subsection{Inverse Problem}

Since image reconstruction in optical tomography is often ill-posed, it is typically formulated as the regularized inversion problem
\begin{equation}
\label{Eq:Optimization}
\xbmhat = \argmin_{\xbm \in \C^N} \left\{g(\xbm)+h(\xbm)\right\},
\end{equation}
where $g$ is the data-fidelity term that ensures the consistency with the measured data, and $h$ is the regularization term that imposes the prior knowledge on the desired image. For example, the Tikhonov regularization~\cite{Tikhonov.Arsening1977} assumes a Gaussian prior on the unknown image. It has been previously used in IDT for deriving a closed form solution~\cite{Ling.etal18}. More recent regularizers, such as the sparsity-promoting $\ell_1$-norm penalty~\cite{Figueiredo.Nowak2001} and the edge-preserving total variation (TV)~\cite{Rudin.etal1992}, are nonsmooth and do not have closed-form solutions, thus requiring iterative algorithms for image formation. In particular, the family of proximal methods---such as proximal gradient method (PGM)~\cite{Figueiredo.Nowak2003, Daubechies.etal2004, Bect.etal2004, Beck.Teboulle2009a} and alternating direction method of multipliers (ADMM)~\cite{Eckstein.Bertsekas1992, Afonso.etal2010, Ng.etal2010, Boyd.etal2011}---avoid the need to differentiate the regularizer by using the proximal map~\cite{Parikh.Boyd2014}.

Recently, deep learning has gained popularity in imaging inverse problems~\cite{Mousavi.etal2015, DJin.etal2017, Kang.etal2017, Ye.etal2018, Sun.etal2018, Li.etal2018, Yoo.etal2019, Goy.etal2019, Nehme.etal18}. Traditional strategy trains the convolutional neural network (CNN) to learn the direct mapping from the measurements to some ground-truth image. Despite their excellent performance in some image reconstruction problems, this strategy does not leverage the known physics of the imaging system and does not insure consistency with the measured data. In this paper, we propose SIMBA to reconcile the model-based and learning-based approaches by infusing deep denoising priors into online iterative algorithms.

\subsection{Regularization by Denoising}

RED~\cite{Romano.etal2017} is a recently introduced framework to leverage powerful image denoisers. It has been successfully applied in many regularized imaging tasks, including image deblurring~\cite{Romano.etal2017}, super-resolution~\cite{Mataev.etal2019}, and phase retrieval~\cite{Metzler.etal2018}. The framework aims to find a fixed point $\xbm^\ast$ that satisfies
\begin{equation}
\label{Eq:FixedPoints}
\Gsf(\xbm^\ast) = \nabla g(\xbm^\ast) + \tau (\xbm^\ast - \Dsf_\sigma(\xbm^\ast)) = 0,
\end{equation}
where $\nabla g$ denotes the gradient of $g$, $\Dsf_\sigma$ is the image denoiser, and $\tau>0$ adjusts the tradeoff between the data-fidelity and the prior. RED algorithms seek a vector $\xbm^\ast$ that lies in the zero set of $\Gsf: \R^n \rightarrow \R^n$
\begin{equation}
\label{Eq:ZeroSet}
\xbm^\ast \in \zer(\Gsf) \defn \{\xbm \in \R^n : \Gsf(\xbm) = 0\}.
\end{equation}
For example, the gradient-method variant of RED (denoted as GM-RED) can be implemented as
\begin{align}
\label{Eq:REDupdate}
\xbm^k &\leftarrow \xbm^{k-1} - \gamma (\nabla g(\xbm^{k-1}) + \Hsf(\xbm^{k-1})) \nonumber \\
&\text{where}\quad \Hsf(\xbm) \defn \tau(\xbm - \Dsf_\sigma(\xbm)).
\end{align}
Here, the parameter $\gamma>0$ is the step-size. When the denoiser $\Dsf_\sigma$ is locally homogeneous and has a symmetric Jacobian~\cite{Romano.etal2017, Reehorst.Schniter2019}, the operator $\Hsf$ corresponds to the gradient of the following regularizer
\begin{equation}
\label{Eq:Regularizer}
h(\xbm) = \frac{\tau}{2}\xbm^\Tsf(\xbm-\Dsf_\sigma(\xbm)).
\end{equation}
By having a closed-form objective function, one can use the classical optimization theory to analyze the convergence of RED algorithms~\cite{Romano.etal2017}. On the other hand, fixed-point convergence has also been established without having an explicit objective function~\cite{Reehorst.Schniter2019, Sun.etal2019a}. Reehorst \emph{et al.}~\cite{Reehorst.Schniter2019} have shown that RED proximal gradient methods (RED-PG) converges to a fixed point by utilizing the monotone operator theory. Sun \emph{et al.}~\cite{Sun.etal2019c} have established the explicit convergence rate for the block coordinate variant of RED (BC-RED) under a nonexpansive $\Dsf_\sigma$. In this paper, we extend these prior analyses to the randomized processing of the measurements instead of image blocks, which opens up applications to tomographic imaging with a large number of projections.

\section{Proposed Method}
\label{Sec:Proposed}
We now introduce SIMBA that combines the iterative usage of the forward model with a deep denoising prior. At each iteration, SIMBA updates $\xbm$ by combining a stochastic gradient for increasing data-consistency with a CNN denoiser for artifact reduction. SIMBA is ideal for data-intensive biomedical imaging applications where the object features are difficult to characterize using traditional regularizers.

\subsection{Iterative Online Procedure}

\begin{figure}[t]
\centering
\begin{minipage}[t]{.5\textwidth}
\begin{algorithm}[H]
\caption{$\mathsf{SIMBA}$}\label{alg:OnRED}
\begin{algorithmic}[1]
\State \textbf{input: } $\xbm^0 \in \R^n$, $\tau > 0$, $\sigma > 0$, and $B \geq 1$
\For{$k = 1, 2, \dots$}
\State $\nablahat g(\xbm^{k-1}) \leftarrow \mathsf{minibatchGradient}(\xbm^{k-1}, B)$
\State $\Gsfhat(\xbm^{k-1}) \leftarrow \nablahat g(\xbm^{k-1}) + \tau(\xbm^{k-1}-\Dsf_\sigma(\xbm^{k-1}))$
\State $\xbm^k \leftarrow \xbm^{k-1} - \gamma \Gsfhat(\xbm^{k-1})$
\EndFor\label{euclidendwhile}
\end{algorithmic}
\end{algorithm}
\end{minipage}
\end{figure}

In IDT, the data-fidelity term can be written as an average over a set of distinct components functions
\begin{equation}
\label{Eq:ComponentData}
g(\xbm)  = \frac{1}{I}\sum_{i = 1}^I g_i(\xbm),
\end{equation}
where each component function $g_i$ is evaluated only on the subset $\ybf_i$ of the full measurements $\ybf$
\begin{align}
\label{Eq:GradientIDT}
g_i(\xbm) = \mathcal{L}(\ybm_i, \Abm_i\xbm)\;,
\end{align}
where $\Lcal$ is a smooth loss function quantifying the discrepancy between the predicted measurement $\Abm_i\xbm$ and the actual measurements $\ybm_i$. For example, all the results in this paper were obtained using
$$g_i(\xbm) = \frac{1}{2}\|\ybm_i - \Abm_i\xbm\|_2^2 \quad\Rightarrow\quad \nabla g_i(\xbm) = \Abm_i^\Hsf(\Abm_i\xbm - \ybm_i).$$ 
The computation of the gradient of $g$
\begin{align}
\label{Eq:ComponentGradient}
&\nabla g(\xbm) = \frac{1}{I}\sum_{i = 1}^I \nabla g_i(\xbm),
\end{align}
is proportional to the total number of illuminations $I$. A large $I$ effectively precludes the applicability of the batch RED algorithms due to the computational cost of evaluating $\nabla g$.

SIMBA, summarized in Algorithm~\ref{alg:OnRED}, improves scalability through partial randomized processing of gradient components $\nabla g_i$ via the following \emph{minibatch} approximation of the gradient
\begin{equation}
\label{Eq:StochGrad}
\nablahat g(\xbm) = \frac{1}{B}\sum_{b = 1}^B \nabla g_{i_b}(\xbm),
\end{equation}
where $i_1, \dots, i_B$ are independent random indices that are distributed uniformly over $\{1, \dots, I\}$. Due to its ability to control the minibatch size $1 \leq B \leq I$, SIMBA benefits from considerable \emph{flexibility} for trading off different practical considerations compared to the batch RED algorithms. For example, one can consider using small minibatches ($B \ll I)$ for problems where $\nabla g$ is computationally expensive and $\Dsf_\sigma$ is relatively efficient. On the other hand, one can consider larger minibatches for problems where $\Dsf_\sigma$ is relatively slow compared to $\nabla g$. In this paper we focus on image denoisers corresponding to deep neural nets, thus obtaining $\Dsf_\sigma$ that is both fast and effective for many practical imaging problems.

\subsection{CNN-based Denoiser}
\label{Sec:CNNDenoisers}

In recent years, CNNs have been shown to achieve the state-of-the-art performance on image denoising~\cite{Zhang.etal2017, Zhang.etal2018}. We propose a simple denoising network $\DnCNNast$ as the deep learning module in \proposed. The architecture of the neural network, illustrated in Figure~\ref{fig:Schema}, is adapted from the popular DnCNN. In general, $\DnCNNast$ consists of two parts. The first part contains $N_\ell-1$ sequential composite convolutional layers, each of which has one convolutional layer followed by a rectified linear unit (ReLU) layer. The second part is a single convolutional layer that outputs the final denoised image, resulting the total number of layers in $\DnCNNast$ to be $N_\ell$. All the convolution filters are implemented with size $3\times3$, and every feature map has 64 channels. In \proposed, we apply this 2D image denoising network to the 3D sample by performing the layer-by-layer denoising along the axial direction $z$.

We generated the training dataset by adding AWGN to the natural images from BSD400 and applying standard data augmentation strategies including flipping, rotating, and rescaling. Note that our training dataset does not include any biomedical image. We employed the residual learning technique~\cite{He.etal2016} in $\DnCNNast$ so that the network is forced to learn the noise residual in the noisy input. $\DnCNNast$ was trained to minimize the following loss
\begin{equation}
\label{Eq:Loss}
\mathcal{L}_\theta= \frac{1}{p} \sum_{i=1}^{p} \left\{\|f_\theta(\xbm_i) - \ybm_i\|_2^2 + \rho\|f_\theta(\xbm_i) - \ybm_i\|_1\right\},
\end{equation}
where $\xbm_i$ is the noisy input, $\ybm_i$ is the noise, and $f_\theta(\xbm)$ represents the noise predicted by the neural network. Eq.~\eqref{Eq:Loss} penalizes both the mean squared error (MSE) and the mean absolute error (MAE) between the estimated noise and the ground truth. A \emph{loss parameter} $\rho>0$ is thus introduced to adjust the tradeoff between the two errors for the best training performance. Our results show that our simple $\DnCNNast$ is competitive with traditional denoisers in terms of the imaging quality.

\begin{figure}[t]
\centering\includegraphics[width=8.5cm]{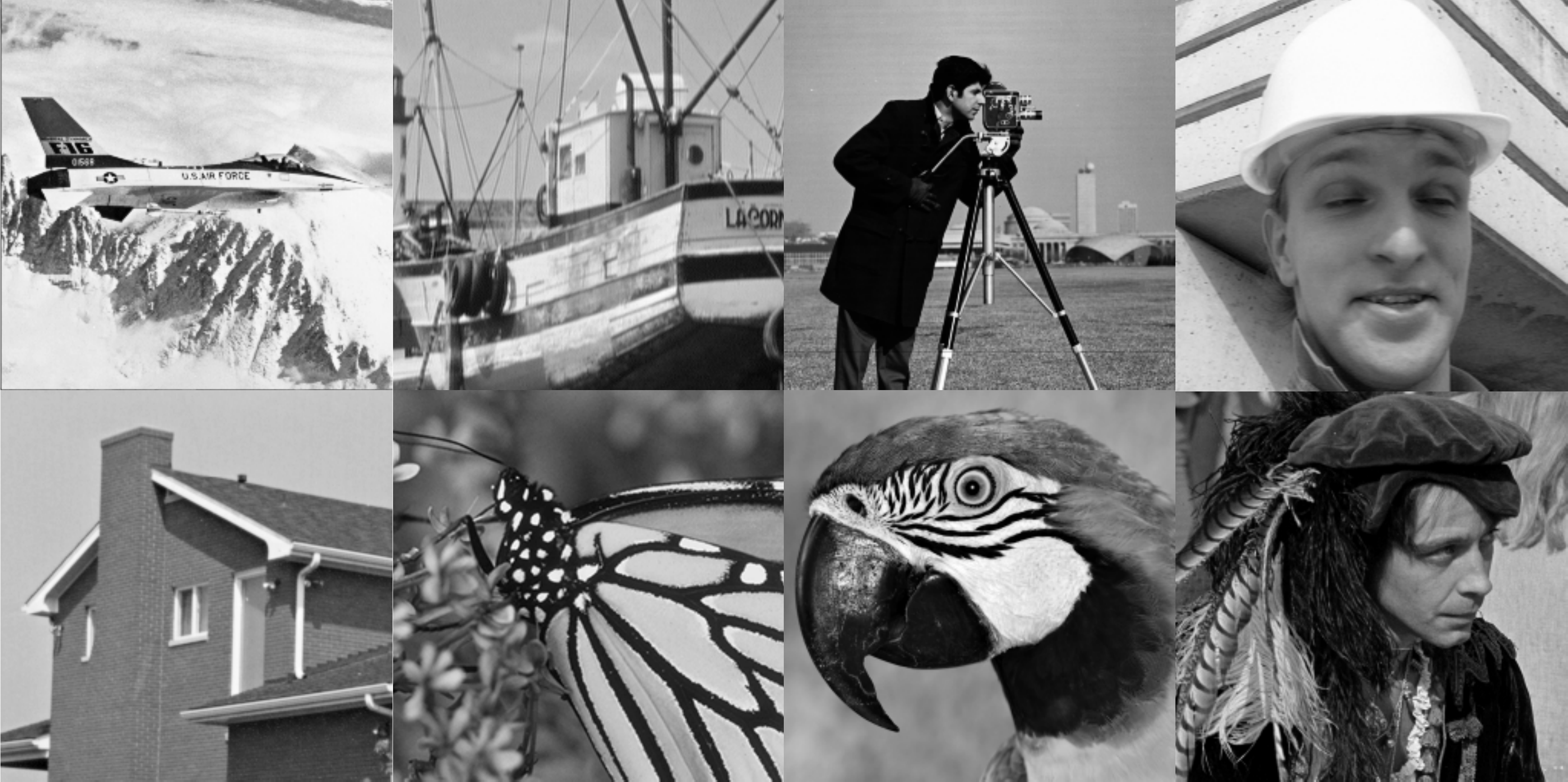}
\caption{Eight test images used in the experiments. Top row from left to right: \textit{Aircraft}, \textit{Boat}, \textit{Cameraman}, \textit{Foreman}. Bottom row from left to right: \textit{House}, \textit{Monarch}, \textit{Parrot}, \textit{Pirate}.}
\label{fig:truth}
\end{figure}

\section{Convergence Analysis}
\label{Sec:Convergence}

Our analysis relies on the fixed-point convergence of averaged operators, which is well known as the Krasnosel'skii-Mann theorem~\cite{Bauschke.Combettes2017}. Here, we extend the result to the iterative online algorithms under the RED formulation and show the worst-case convergence rates. Note that our analysis does not assume that the denoiser corresponds to any explicit RED regularizer. We first introduce the assumptions necessary for our analysis and then present the main results.

\begin{assumption}
\label{As:DataFitConvexity}
We make the following assumptions on the data-fidelity term $g$:
\setlist{nolistsep,leftmargin=*}
\begin{enumerate}[label=(\alph*)]
\item The component functions $g_i$ are all convex and differentiable with the same Lipschitz constant $L > 0$.
\item At every iteration, the gradient estimate is unbiased and has a bounded variance:
$$\E\left[\nablahat g(\xbm)\right] = \nabla g(\xbm),\;\; \E\left[\left\|\nabla g(\xbm)-\nablahat g(\xbm)\right\|_2^2\right] \leq \frac{\nu^2}{B},$$
for some constant $\nu > 0$.
\end{enumerate}
\end{assumption}

\noindent
Assumption~\ref{As:DataFitConvexity} (a) implies that the overall data-fidelity $g$ is also convex and has Lipschitz continuous gradient with constant $L$. Assumption~\ref{As:DataFitConvexity} (b) assumes that the minibatch gradient is an unbiased estimate of the full gradient. The bounded
variance assumption is a standard assumption used in the analysis of online and stochastic algorithms~\cite{Ghadimi.Lan2016, Bernstein.etal2018, Xu.etal2018}

\begin{assumption}
\label{As:NonemptySet}
The operator $\Gsf$ is such that $\zer(\Gsf) \neq \varnothing$.
\end{assumption}
\noindent
Assimption~\ref{As:NonemptySet} is a mild assumption that simply asserts the existence of a solution to~\eqref{Eq:FixedPoints}. It is related to the existence of stationary points in traditional smooth optimization~\cite{Nesterov2004}.

\begin{assumption}
\label{As:NonexpansiveDen}
Given $\sigma>0$, the denoiser $\Dsf_\sigma$ is a nonexpansive operator such that
$$\| \Dsf_\sigma(\xbm) -  \Dsf_\sigma(\ybm)\|_2 \leq \|\xbm-\ybm\|_2\quad \xbm,\ybm \in \R^n,$$
\end{assumption}

\noindent
Nonexpansive variants of several widely used denoisers, such as NLM and DnCNN, have been developed in~\cite{Sreehari.etal2016, Teodoro.etal2019, Ryu.etal2019,Sun.etal2019c}. Under the above assumptions, we can establish the following for \proposed.

\begin{table}[!t]
	\centering
	\scriptsize
	\textbf{\caption{\label{tab:optical} List of parameters of the experimental setup}}
	\begin{tabular*}{8.8cm}{C{1pt}C{100pt}C{52pt}C{54pt}} 			
		\toprule
		\multicolumn{2}{c}{\textbf{Experimental parameters}} & Simulations (\ref{sec:sim}) & Experiments (\ref{sec:exp})\\	 
		\cmidrule(lr){1-2} \cmidrule(lr){3-4}
		$\lambda$   & wavelength of LED light & \SI{630}{\nano\metre}  & \SI{630}{\nano\metre} \\
		$\epsilon_{b}$ & background medium index & $1.33$ & $1.33$ \\
		$z_\text{LED}$ & axial position of LEDs& \SI{-70}{\milli\metre} & \SI{-79}{\milli\metre}\\ 
		$z$ & axial position of the sample& \SI{0}{\micro\metre} & $(-20,100)$~\SI{}{\micro\metre}\\ 
		MO   & microscope objectives & 40$\times$ & 10$\times$ \\
		NA & numerical aperture & 0.65 & 0.25 \\
		\bottomrule
	\end{tabular*}
\end{table}

\begin{table}[!t]
	\centering
	\scriptsize
	\textbf{\caption{\label{tab:param} List of algorithmic hyperparameters}}
	\begin{tabular*}{8.8cm}{C{1pt}C{100pt}C{52pt}C{54pt}} 			
		\toprule
		\multicolumn{2}{c}{\textbf{Hyperparameters}} & Simulations (\ref{sec:sim}) & Experiments (\ref{sec:exp})\\	 
		\cmidrule(lr){1-2} \cmidrule(lr){3-4}
		$\xbm^{0}$   & initial point of reconstructions & $\zerobm$  & $\zerobm$ \\
		$B$ & minibatch size & $20$ & $10$ \\
		$I$ & batch size & $60$ & $89$\\ 
		$\gamma$   & step size & $\frac{1}{L+2\tau}$ & $\frac{1}{L+2\tau}$ \\
		$\sigma$ & input noise level for $\DnCNNast$ & $10$ & $5$ \\
		$\rho$ & loss function parameter & $0$ & $1$ \\
		$N_\ell$ & number of layers in $\DnCNNast$ & $7$ & $10$ \\
		\multirow{2}{*}{$\tau$} & level of regularization in GM-RED and \proposed & optimized for each image & optimized for the dataset\\ 
		\bottomrule
	\end{tabular*}
\end{table}

\begin{figure}[t]
\centering\includegraphics[width=8.5cm]{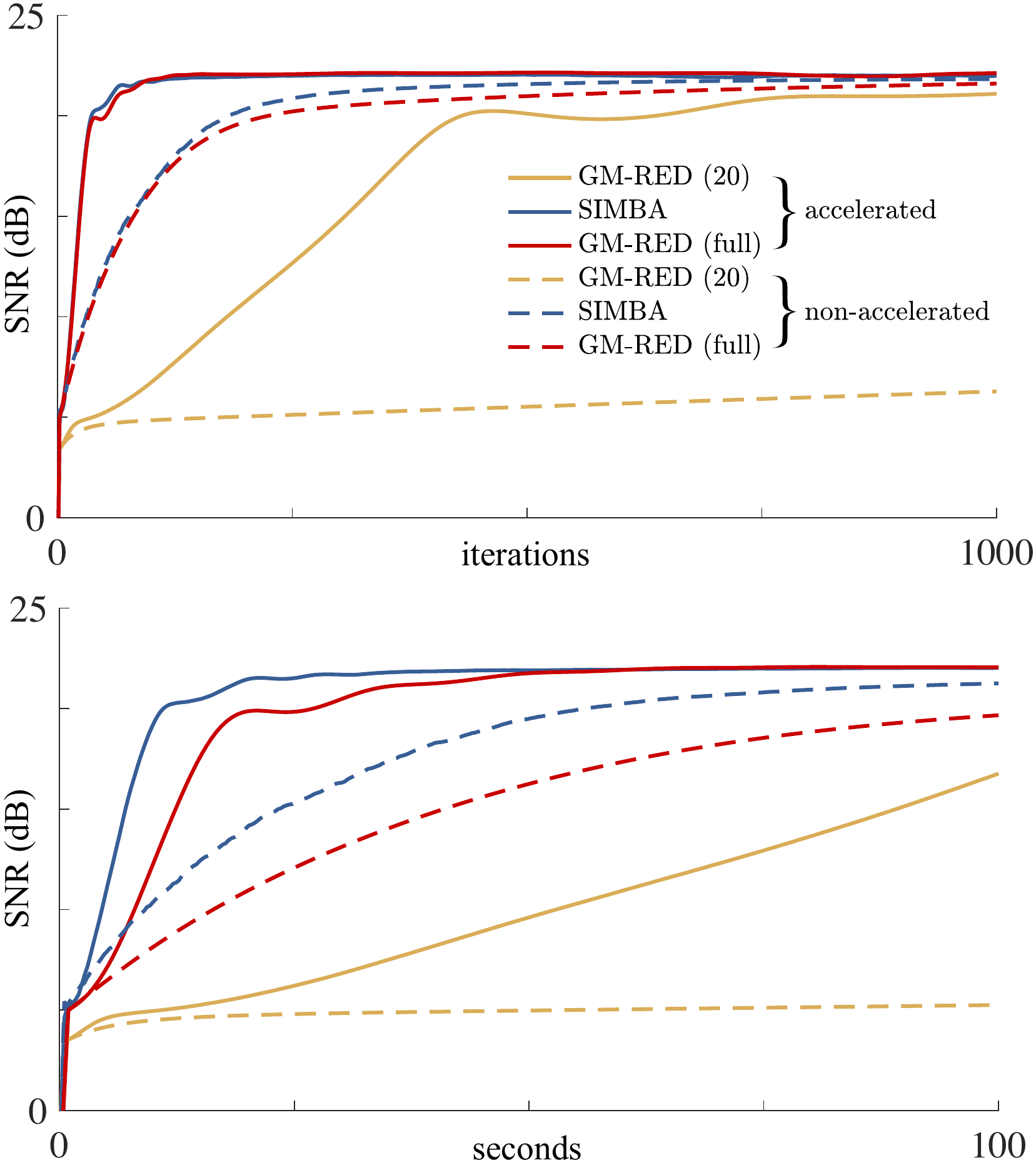}
\caption{Illustration of convergence in SNR of \proposed~with minibatch size $B=20$ under the $\DnCNNast$ denoiser. The top and bottom figures plot the SNR values against the number of itartions and running time, respectively. Two batch algorithms, GM-RED (20) and GM-RED (full), are plotted for comparison. Under the same per-iteration complexity, \proposed~converges to significantly higher SNR than GM-RED (20) due to its actual usage of the full data. Moreover, online processing makes \proposed~converge significantly faster than GM-RED (full). The acceleration is due to the lower computational cost of processing a small random subset of the full data. The same trend is observed for both accelerated and normal versions of the algorithms.}
\label{fig:SNR}
\end{figure}

\begin{table*}[t]
	\centering
	\scriptsize
	\textbf{\caption{\label{tab:SNR} Optimized SNR for each test image in dB}}
	\begin{tabular*}{405pt}{L{28pt}C{28pt}C{35pt}C{28pt}C{28pt}C{28pt}C{28pt}C{28pt}C{28pt}C{28pt}} 	
		\toprule
		\textbf{Algorithms} & GM (20) & SGM & GM (full) & \multicolumn{2}{c}{GM-RED (20)} & \multicolumn{2}{c}{\proposed} & \multicolumn{2}{c}{GM-RED (full)}\\
		\cmidrule(lr){2-2} \cmidrule(lr){3-3} \cmidrule(lr){4-4} \cmidrule(lr){5-6} \cmidrule(lr){7-8} \cmidrule(lr){9-10}
		\textbf{Denoisers} & --- & --- & --- & BM3D & DnCNN$^\ast$ & BM3D & DnCNN$^\ast$ & BM3D & DnCNN$^\ast$   \\
		\midrule
		\textit{Aircraft} & 17.44     & 18.00     & 18.01     & 18.82       & 19.85        & 19.65 		&20.48        &19.47         & 20.44     \\
		\textit{Boat}    & 18.09     & 18.78     & 18.82     & 19.96       & 20.40        & 21.23		&21.29        &21.12         & 21.55     \\
		\textit{Cameraman}   & 16.27    & 17.06   & 17.08     & 17.36     & 18.83       & 18.36        & 19.59         & 18.34        & 19.32     \\
		\textit{Foreman} & 22.78    & 23.81     & 23.88     & 25.24     & 27.35       & 26.76        & 28.55     & 26.78        & 28.71		\\
		\textit{House}  & 19.89    & 20.73     & 20.79     & 22.02     & 22.57       & 22.88        & 23.47     & 22.86        & 23.56		\\
		\textit{Monarch}  & 17.70    & 18.65     & 18.69     & 19.16     & 21.50       & 20.37        & 22.63     & 20.42        & 22.81		\\
		\textit{Parrot}  & 17.85    & 18.69     & 18.72     & 18.38     & 19.22       & 19.11        & 19.79     & 19.14        & 20.05		\\
		\textit{Pirate}  & 19.17    & 19.79     & 19.81     & 19.76     & 20.11       & 20.41        & 20.98     & 20.44        & 21.06		\\
		\midrule
		\textbf{Average} & 18.65     & 19.44     & 19.48     & 20.09       & 21.23        & 21.10 	& 22.10       & 21.07        & 22.19	   \\
		\bottomrule
	\end{tabular*}
\end{table*}

\begin{theorem}
\label{Thm:ConvThm1}
Run \proposed~for $t \geq 1$ iterations under Assumptions~\ref{As:DataFitConvexity}-\ref{As:NonexpansiveDen} using a fixed step-size $0 < \gamma \leq 1/(L+2\tau)$ and a fixed minibatch size $B\geq1$. Then, we have
\begin{align}
\E\left[\frac{1}{t}\sum_{k = 1}^t \|\Gsf(\xbm^{k-1})\|_2^2\right]  \leq(L+2\tau) \left[\frac{\|\xbm^0-\xbmast\|_2^2}{\gamma t} + \frac{\gamma\nu^2}{B} \right]. \nonumber
\end{align}
\end{theorem}

\begin{proof}
See Appendix~\ref{Sec:Proof}.
\end{proof}
\noindent
Thereom~\ref{Thm:ConvThm1} is analogous to the convergence of the \emph{minibatch stochastic gradient descent (SGD)}~\cite{Robbins.Monro1951} in the sense that the convergence can be established up to an error term that depends on $\gamma$ and $B$. Similarly, the accuracy of the expected convergence of SIMBA to $\zer(\Gsf)$ improves with smaller $\gamma$ and larger $B$. For example, by setting ${\gamma = 1/[(L+2\tau)\sqrt{t}]}$, we get
$$\E \left[\min_{k \in \{1, \dots, t\}} \|\Gsf(\xbm^{k-1})\|_2^2\right] \leq \E\left[\frac{1}{t}\sum_{k = 1}^t \|\Gsf(\xbm^{k-1})\|_2^2\right]\leq \frac{C}{\sqrt{t}},$$
where $C$ is some positive constant.

Finally, note that the analysis in Theorem~\ref{Thm:ConvThm1} only provides \emph{sufficient conditions} for the convergence of SIMBA. As corroborated by our numerical studies in Section~\ref{Sec:Experiments}, the actual convergence of SIMBA is more general and often holds beyond nonexpansive denoisers (such as BM4D). One plausible explanation for this is that such denoisers are \emph{locally nonexpansive} over the set of input vectors used in testing (see also discussion in~\cite{Ryu.etal2019}). On the other hand, the recent techniques for spectral-normalization of deep neural nets~\cite{Miyato.etal2018, Sedghi.etal2019, Gouk.etal2018} provide a convenient tool for building \emph{globally nonexpansive} neural denoisers that result in provable convergence of SIMBA.

\begin{figure}[t]
\centering\includegraphics[width=8.5cm]{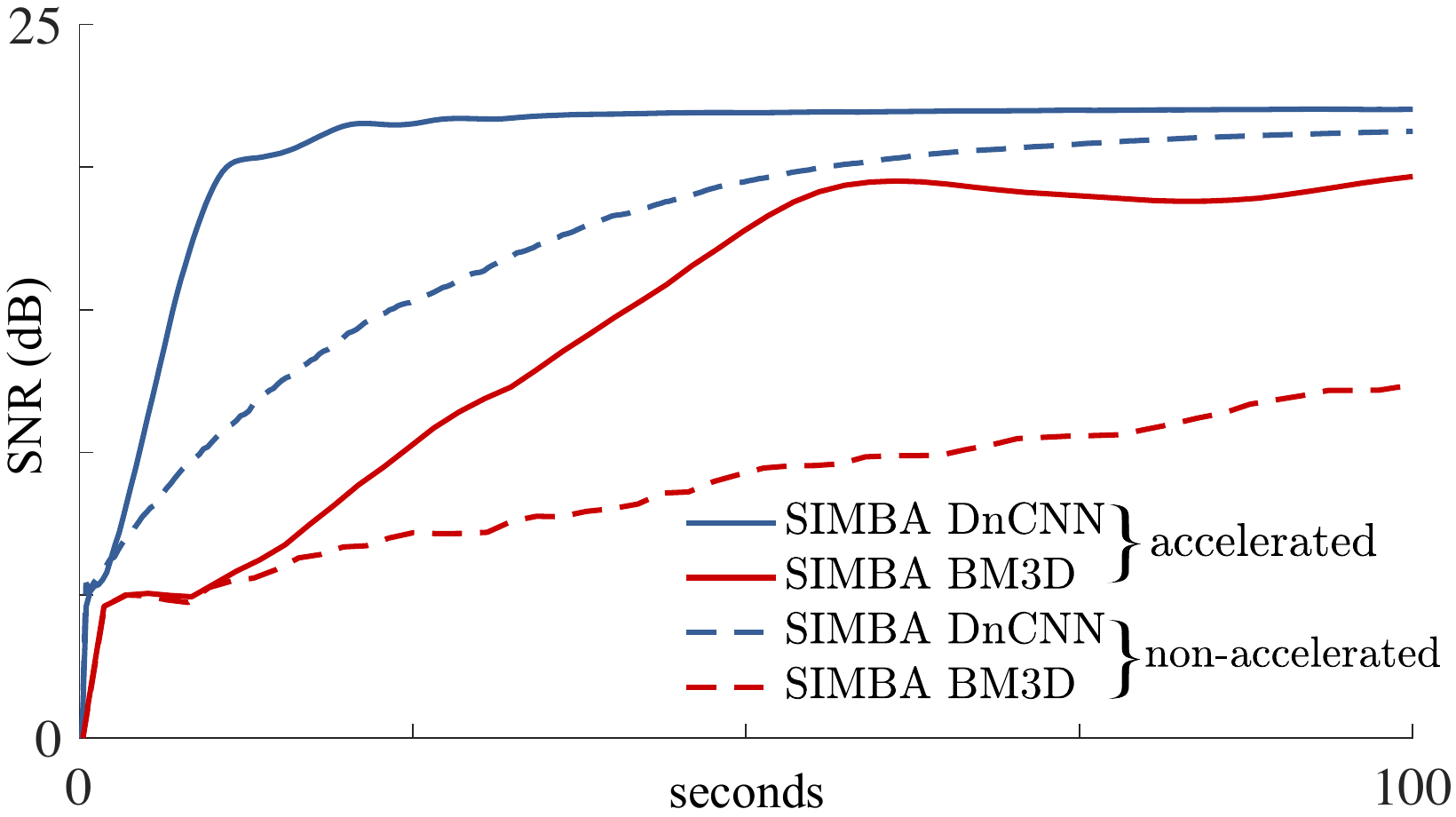}
\caption{Illustration of the convergence speed of SIMBA with minibatch size $B=20$ under the $\DnCNNast$ and BM3D denoisers. The figure plots the SNR values against the running time in seconds. SIMBA DnCNN is significantly faster due to the fast GPU implementation of the denoiser.}
\label{fig:Comparison}
\end{figure}

\begin{figure*}[t]
	\centering\includegraphics[width=\linewidth]{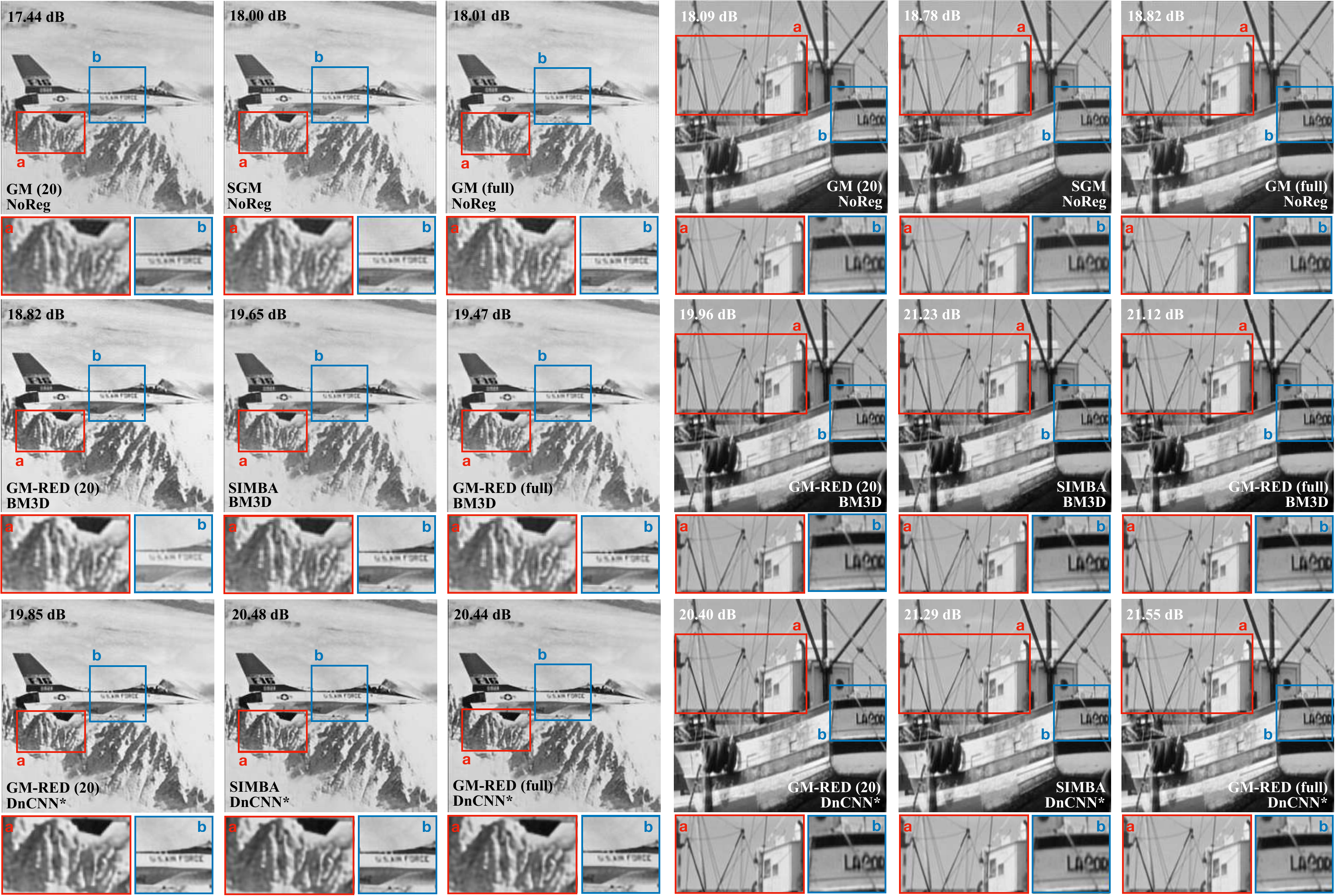}
\caption{Visual examples of reconstructed \textit{Aircrafts} (left) and \textit{Boat} (right) images by different algorithms. Three columns correspond to algorithms using fixed 20, random 20 out of 60, and full 60 measurements, respectively. The first row presents the unregularized results and the second and third row show the results given by a well-known BM3D denoiser and a state-of-the-art deep learning prior, respectively. Differences are zoomed in using boxes inside the images. Each image is labeled by its SNR (dB) with respect to the original image. Note that our proposed algorithm \proposed~recovers the details lost by the batch algorithm with the same computational cost and achieves the same level of SNR and visual quality as the full batch algorithm.}
\label{fig:examples}	
\end{figure*}

\section{Experimental Validation}
\label{Sec:Experiments}

In this section, we validate \proposed~on both simulated and experimental data. We first numerically demonstrate the efficiency and practical convergence of \proposed~in simulations. Next, we apply \proposed~to reconstruct a 3D model from a set of real intensity-only measurements. Our results highlight the applicability and effectiveness of \proposed~for the iterative inversion in optical tomography.

\subsection{Setup}

In simulations, we reconstruct eight grayscale natural images, representing the phase component of the complex permittivity contrast, displayed in Figure \ref{fig:truth}. They are assumed to be on the focal plane $z=0$\,\SI{}{\micro\meter} with LEDs located at $z_\text{LED}=$~\SI{-70}{\milli\meter}. We generate $I=60$ simulated intensity measurements with 40$\times$ microscope objectives (MO) and 0.65 numerical aperture (NA). All simulated measurements are corrupted by AWGN corresponding to 20 dB of \emph{input signal-to-noise ratio (SNR)}. As a quantitative metric for measuring the quality of reconstructions, we use the SNR defined as follows 
$$\operatorname{SNR}(\hat{\ybm}, \ybm) \triangleq \max _{a, b \in \mathbb{R}}\left\{20 \operatorname{log} _{10}\left(\frac{\|\ybm\|_{\ell_{2}}}{\|\ybm-a \hat{\ybm}+b\|_{\ell_{2}}}\right)\right\}$$
where $\hat{\ybm}$ represents the noisy vector and $\ybm$ denotes the ground truth. In experiments, we recover a 3D algae sample from real IDT measurements. The 3D sample is located over the range $(-20,100)$~\SI{}{\micro\metre} and  $z_\text{LED}=$~\SI{-79}{\milli\meter}. We set the slice spacing as \SI{5}{\micro\metre}, so each slice represents the average over the sample thickness. We take $I=89$ measurements with 10$\times$ MO and 0.25 NA for reconstruction. We refer to Table \ref{tab:optical} for the detailed summary of the experimental parameters. All experiments in this paper were performed on a machine equipped with an Intel Xeon E5-2620 v4 Processor that has 4 cores of 2.1 GHz and 256 GBs of DDR memory. We trained all neural nets using NVIDIA RTX 2080 GPUs.

The algorithmic hyperparameters are summarized in Table \ref{tab:param}. All algorithms start from $\xbm^{0}=\zerobm$. We trained $\DnCNNast$ for the removal of AWGN at four noise levels corresponding to $\sigma \in \{5, 10, 15, 20\}$. The same set of $\sigma$ is used for BM3D. All algorithmic parameters are optimized for the best performance. 

\begin{figure*}[t]
	\centering\includegraphics[width=0.9\linewidth]{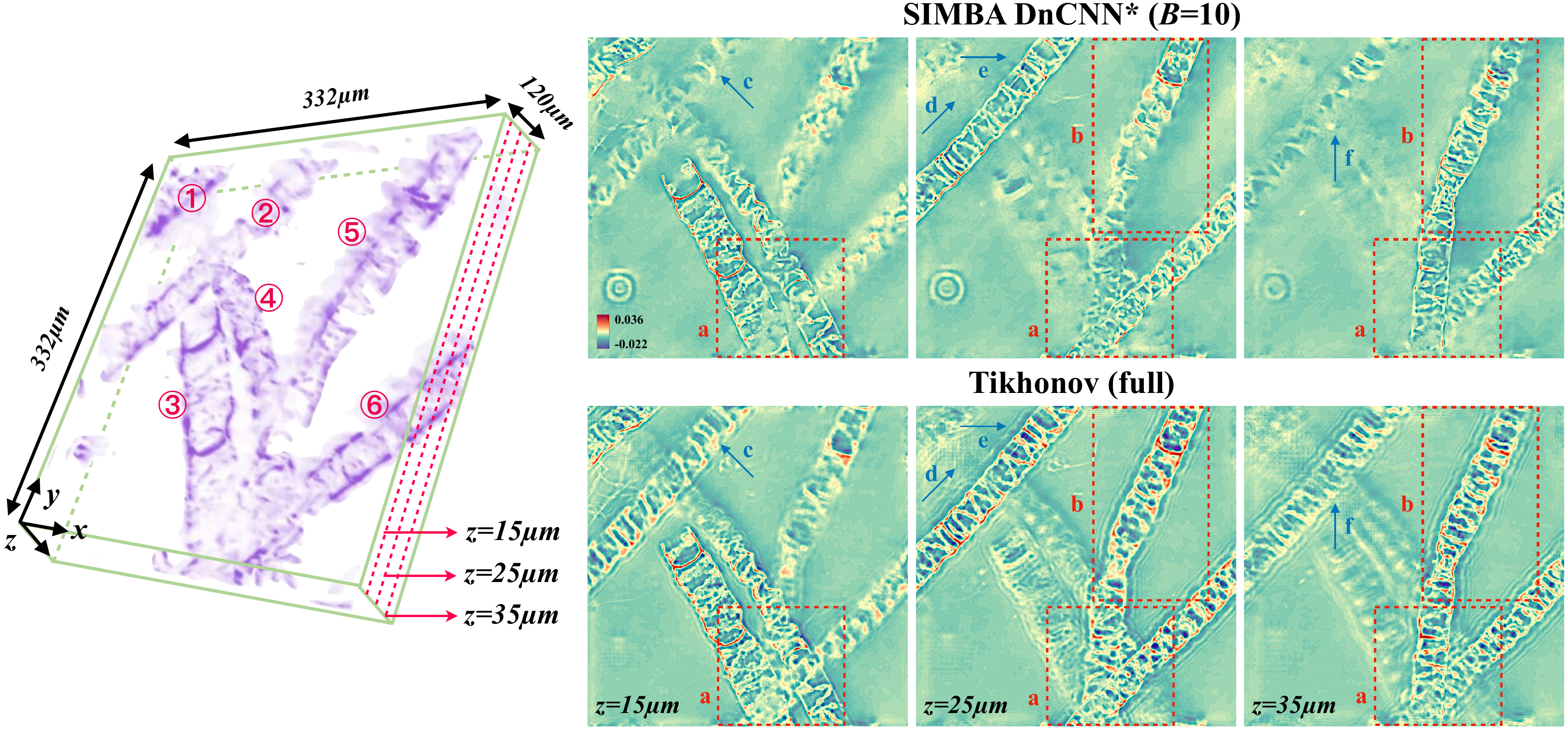}
	\caption{Visualization of the 3D algae reconstruction. Algae are labeled by circled numbers. We select three slices of the sample to illustrate the improvement of performance by our proposed method \proposed~with $B=10$ over Tikhonov (full), which uses all 89 measurements. Regions (a) and (b) demonstrate the better axial sectioning effect of \proposed~and arrows (c) to (f) point out the areas where \proposed~suppresses the artifacts present in the Tikhonov reconstruction.}
	\label{fig:3D}
\end{figure*}

\begin{figure*}[h]
	\centering\includegraphics[width=0.9\linewidth]{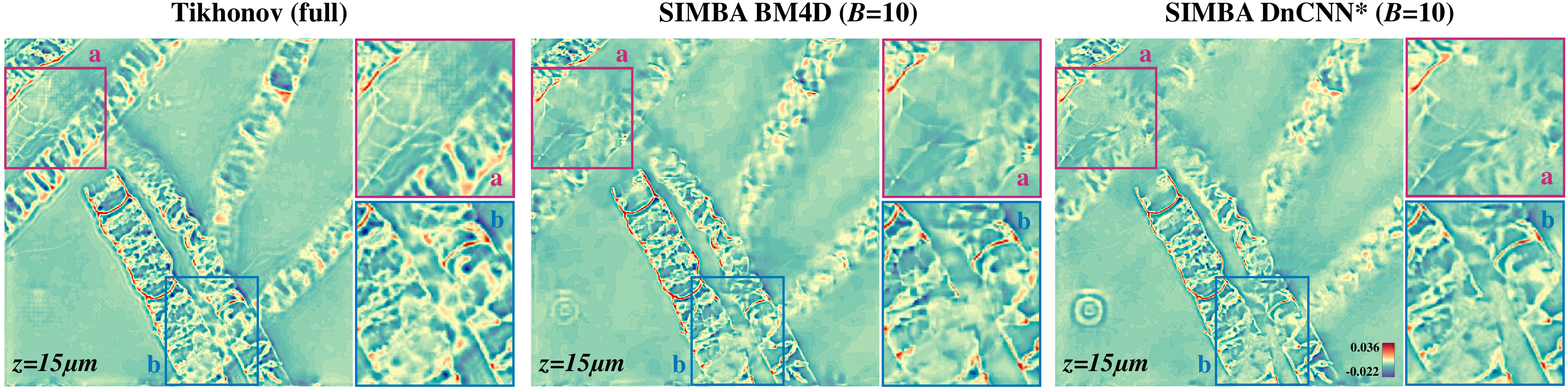}
	\caption{Comparison of \proposed~under BM4D and $\DnCNNast$ against Tikhonov. Regions (a) and (b) are zoomed in to highlight visual differences. Tikhonov reconstructed image contains grid-shape artifacts and interfering contents from other slices, while BM4D generates blocks and nonsmoothness. \proposed~under $\DnCNNast$ produces the most real recovery with the clearest shape of the algae.}
	\label{fig:Algae}
\end{figure*}

\subsection{Simulated Data}
\label{sec:sim}

In this section, we numerically illustrate the advantages of \proposed~in tomographic imaging over the batch GM-RED. The advantages are: (1) better SNR under a limited memory budget; (2) better time efficiency when all the measurements are used.

\begin{figure*}[t]
	\centering\includegraphics[width=0.90\linewidth]{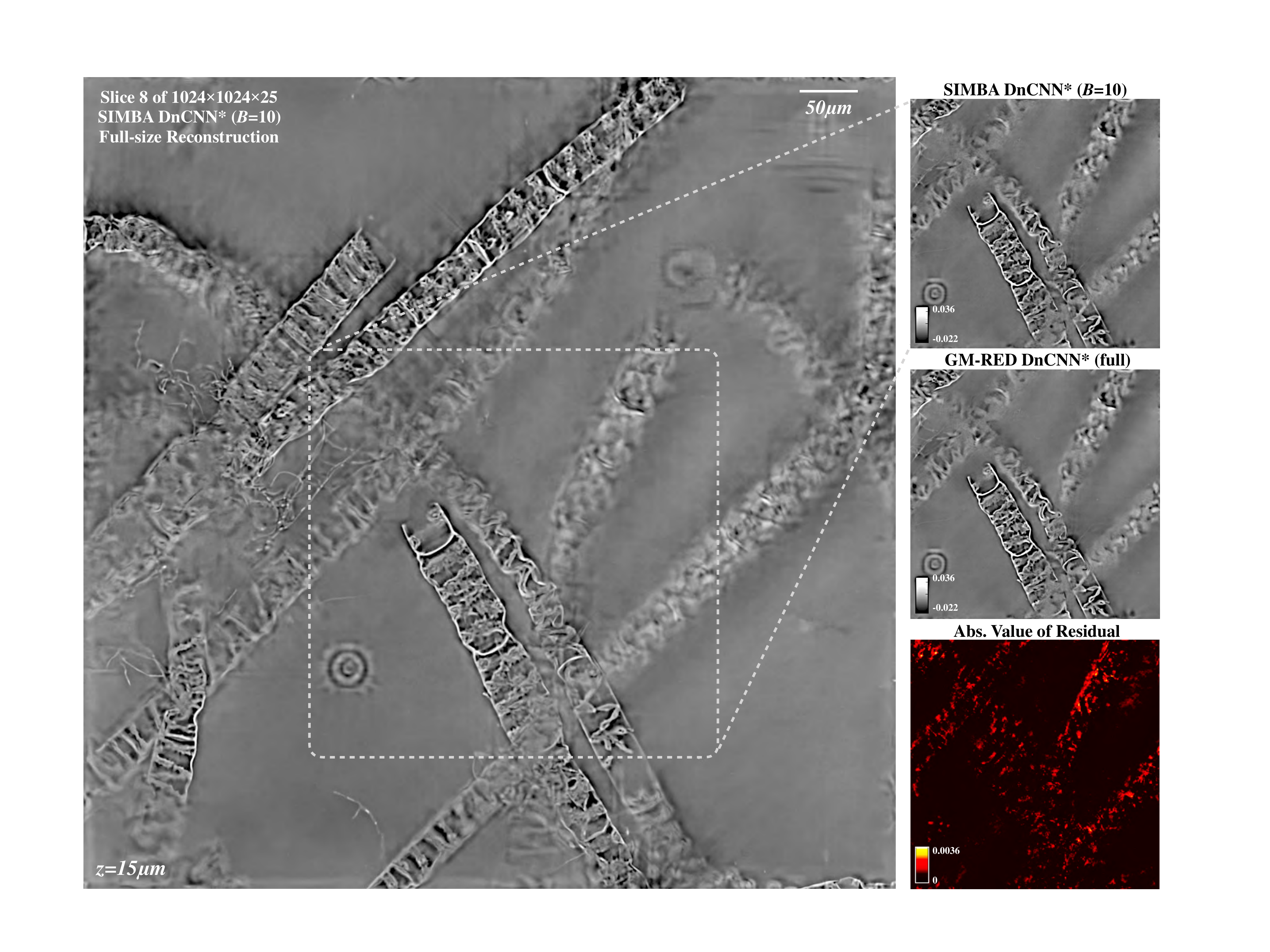}
	\caption{A slice from the full $1024\times1024\times25$ reconstruction by \proposed~under $\DnCNNast$. On the right is a comparison between \proposed~and the full batch results for the dotted region. The two reconstructions are visually indistinguishable, and the absolute value of the residual between them highlights the numerical proximity of \proposed~to the full batch reconstruction. Note the small numerical scale of the residual compared to that of the two reconstructions.}
	\label{fig:residual}
\end{figure*}

\renewcommand{\arraystretch}{1.1}

\begin{table*}[t]
	\centering\scriptsize
	\textbf{\caption{\label{Tab:Memo}Per-iteration memory usage specification for processing our experimental data}}
	\begin{tabular*}{392pt}{R{15pt}C{30pt}C{40pt}C{71pt}C{40pt}C{71pt}C{40pt}C{40pt}} 	
		\toprule
		\multicolumn{2}{c}{\textbf{Algorithms}} & & \multicolumn{2}{c}{\proposed~(10)} & \multicolumn{2}{c}{GM-RED (full)} & \\
		\cmidrule(lr){4-5} \cmidrule(lr){6-7}
		\multicolumn{2}{c}{\textbf{Variables}} & Data type &  Shape & Size & Shape & Size   \\
		\midrule
		\multirow{2}{*}{$\{\Abm_i\}$}   & phase   & complex64 & $1024\times1024\times25\times10$ & 3.91 Gb & $1024\times1024\times25\times89$ & 34.77 Gb   \\[0.5ex]
		 & absorption   &  complex64  & $1024\times1024\times25\times10$ & 3.91 Gb & $1024\times1024\times25\times89$ & 34.77 Gb  \\[0.5ex]
		\multicolumn{2}{c}{$\{\ybm_i\}$}  &  complex128   & $1024\times1024\times10$  & 0.31 Gb & $1024\times1024\times89$  & 2.78 Gb  \\[0.5ex]
		\multicolumn{2}{c}{others combined} & ---  & ---  & 3.13 Gb & --- & 3.13 Gb \\
		\midrule
		\multicolumn{2}{c}{\textbf{Total}} & ---  & ---  & \colorbox{lightgreen}{\textbf{11.26 Gb}} & ---  & 75.45 Gb  \\
		\bottomrule
	\end{tabular*}
\end{table*}

Figure~\ref{fig:SNR} (top) plots the average SNR over test images against the iteration number for \proposed~and GM-RED (20), both using $\DnCNNast$ as the denoiser. GM-RED (20) uses a fixed set of 20 (out of 60) measurements, while \proposed~selects a random subset of 20 at every iteration. Under the same computational complexity, \proposed~achieves a SNR boost of about 1 dB over GM-RED (20) because the former has access to all the measurements. Visual examples are presented in Figure \ref{fig:examples}. 
As a reference, we also plot the SNR for GM-RED using all 60 measurements, denoted as GM-RED (full).

Figure~\ref{fig:SNR} (bottom) highlights the faster time convergence of \proposed~compared to GM-RED (full) to the same level of SNR. Figure~\ref{fig:examples} highlights that the SNR values and the visual quality obtained by \proposed~and GM-RED (full) are nearly identical. \proposed,~however, significantly reduces the reconstruction time by processing one third of all measurements at each iteration. Specifically, the average per-iteration times of GM-RED (20), \proposed, and GM-RED (full) are 0.30 second, 0.31 second, and 0.52 second, respectively. We also note that by processing only a subset of measurements, \proposed~leads to a more favorable tradeoff between computational cost and memory compared to GM-RED (full). This makes \proposed~beneficial for processing datasets containing a large number of tomographic measurements. The convergence speed of SIMBA can be significantly improved by using deep learning denoisers implemented on GPUs. This is highlighted in Figure~\ref{fig:Comparison}, where SIMBA $\DnCNNast$ is compared against SIMBA BM3D.

Table \ref{tab:SNR} shows final SNRs of all reconstructions we performed. We run all simulations using the accelerated versions of these algorithms, which are analogous to the accelerated gradient method by Nesterov~\cite{Nesterov2004}. Empirically, they converge to the same solution as the non-accelerated counterparts. For reference, we show the evolution of SNR for non-accelerated versions by the dotted lines in Figure \ref{fig:SNR}. Table \ref{tab:SNR} shows that our DnCNN denoiser has higher average SNR than BM3D. The compatibility of \proposed~with $\DnCNNast$, which is a low-complexity denoiser, increases the potential of applying \proposed~to large scale image reconstructions. 

\subsection{Experimental IDT Dataset}
\label{sec:exp}

In this section, we use \proposed~to reconstruct a 3D algae sample of $1024\times1024\times25$ pixels from 89 high-resolution measurements. The large sample volume dramatically increase the memory usage and computational cost, and prohibits the applicability of the full batch algorithms. Experimental results show that \proposed~successfully overcomes these difficulties by processing a small subset of all measurements ($B=10$) at every iteration and leads to significant performance improvements compared to the method reported in~\cite{Ling.etal18}.

Figure \ref{fig:3D} provides a 3D visualization of the phase component of the image recovered by SIMBA, with different algae labeled by circled numbers (there are 6 of them). Figure~\ref{fig:Algae} compares three slices of our SIMBA results and the Tikhonov (full) results obtained by algorithm in \cite{Ling.etal18}, which uses all 89 measurements. As discussed in~\cite{Ling.etal18}, the DC component of the phase is lost in the IDT forward model, we thus set the mean of all the results to the one of the Tikhonov reconstruction for a more uniform comparison. We evaluate the quality of different reconstructions by comparing their axial sectioning effect and the ability to eliminates artifacts. In the 3D tomographic model with strong sectioning effect, a pattern emerges only in the slice it belongs to and fades away as we go axially to different depths. Sectioning enables us to better predict the axial location of the patterns within a 3D object and thus better understand its internal structure, which is crucial for biomedical imaging applications. While Tikhonov regularization is attractive from computational perspective, it is known that to lead to excessive smoothing. This complicates the understanding of the axial structure of the sample. On the other hand, by leveraging the $\DnCNNast$ prior, \proposed~improves the performance, while also mitigating the computational complexity with online processing. Our results show that \proposed~with $\DnCNNast$ enables better sectioning of the object compared to the Tikhonov prior. For example, maintaining the clarity and sharpness of algae \ctwo~in slice $z=$\,\SI{25}{\micro\meter}, \proposed~successfully reduces the artifacts generated by the content of adjacent slices, which exist in the region (a) of Tikhonov. In the other two slices, algae \ctwo~fades away and does not generate strong shadowy artifacts as indicated by arrows (c) and (f). By horizontally comparing the two rows, the algae cluster in region (a) is visually better resolved by \proposed~than Tikhonov. Moreover, in \proposed~reconstructions, the top half of algae \cfive~in region (b) looks sharp in slice $z=$\,\SI{25}{\micro\meter} and the bottom half appears clear in slice $z=$\,\SI{35}{\micro\meter}. This inter-slice information implies that algae \cfive~penetrate through $z=$\,\SI{25}{\micro\meter} and $z=$\,\SI{35}{\micro\meter}. However, the whole structure of algae \cfive~is present in both slices of Tikhonov reconstructions, which fails in illustrating the axial position. Note that \proposed~also better eliminates artifacts pointed out by arrows (d) and (e). To further analyze the performance of the priors, we bring BM4D, the 3D version of the well-known denoiser BM3D, into comparison. In zoom-in region (a) of Figure \ref{fig:Algae}, Tikhonov reconstruction contains grid-shape artifacts. BM4D generates small blocks due to its block-matching mechanism. $\DnCNNast$ provides a more real and sharper result than the other two. In region (b), Tikhonov reconstruction is of satisfactory visual quality but the shadow of algae \cfive~and \csix~in the background interferes with the actual content in this slice. BM4D erases the shadow in the background but it again generates blocky artifacts which makes its reconstruction not as real as $\DnCNNast$ result.

Finally, we present one slice of the full $1024\times1024\times25$ reconstruction by \proposed~under $\DnCNNast$ in Figure \ref{fig:residual}. For comparison, we run GM-RED (full) under $\DnCNNast$ but only for the dotted region because of the high computational cost of the full batch reconstruction. The result is juxtaposed with our \proposed~result. These two algorithms are run with the same $\tau$ value until convergence. Visually, they look almost identical and we present the absolute value of the residual between the two for reference. The residual is negligible compared to the numerical scale of the two results. Quantitatively, if we assume the result of the full batch algorithm to be the ``ground truth", the SNR of \proposed~is 47.03 dB. This substantiates that \proposed~sufficiently matches the full batch algorithms in terms of the final reconstruction quality. Specifically, the average per-iteration running time of \proposed~for reconstructing the dotted region is 22 seconds, while that of GM-RED is 192 seconds, which corresponds to a $9\times$ speed-up. SIMBA also requires less memory at every iteration by processing only about one ninths of full measurements. The reduced running time and memory usage in processing such an intensive amount of data highlighted the efficiency improvement of \proposed~compared to the traditional batch GM-RED. 

We would also like to note that the memory considerations in image reconstruction must take into account the size of all the variables related to the image volume $\xbm$, the measured data $\{\ybm_i\}$, and the measurement operators $\{\Abm_i\}$. 
The goal of this paper is to address the problems where the bottleneck is in the storage and processing of the measurements and measurement operators. 
Table~\ref{Tab:Memo} records the total memory (Gb) used by SIMBA and GM-RED (full) in each iteration. \textcolor{red}{Note that our implementation stores each $\Abm_i$ as two separate matrices for phase and absorption. Additionally, each matrix is stored in the Fourier space to reduce computational complexity of computing convolutions. This results in the storage of complex valued arrays for each, consisting of pairs of double precision floats for every element.} While GM-RED (full) requires 75.45 Gb of memory due to its processing of all measurements in every iteration, SIMBA uses only 11.26 Gb, about one-seventh of the full volume, which makes the algorithm particularly well suited to tomographic applications where one needs to process a very large number of views.

\section{Conclusion}
\label{Sec:Conclusion}
We proposed an extension of RED for solving imaging inverse problems in optical tomography. Our method is scalable to large measurements and uses a deep denoising prior to improve the final estimate. We proved the fixed-point convergence of the method without assuming an explicit objective function, which complements the current theoretical analysis of RED for large-scale image reconstruction. We validated the method on both simulated and experimental IDT data. Especially, the 3D reconstruction of a large algae sample fully elucidates the benefits of our method in data-intensive imaging problems. Future work includes the application of \proposed~in other advanced IDT modalities with coded illumination patterns~\cite{Matlock.etal2019} and accelerated data acquisition~\cite{Li.etal2019}.

\section*{Acknowledgements}

This work was partially supported by grants NSF CCF-1813910.

\bibliographystyle{IEEEtran}


\newpage

\appendix

\section{Background Material}

The results in this section are well-known in the optimization literature and can be found in different forms in standard textbooks~\cite{Rockafellar.Wets1998, Boyd.Vandenberghe2004, Nesterov2004, Bauschke.Combettes2017}. For completeness, we summarize the key results useful for our analysis. 

\medskip
\begin{definition}
An operator $\Tsf$ is Lipschitz continuous with constant $\lambda > 0$ if
$$\|\Tsf\xbm - \Tsf\ybm\| \leq \lambda\|\xbm-\ybm\|,\quad \xbm, \ybm \in \R^n.$$
When $\lambda = 1$, we say that $\Tsf$ is nonexpansive. 
\end{definition}

\begin{definition}
$\Tsf$ is cocoercive with constant $\beta > 0$ if
$$(\Tsf\xbm-\Tsf\ybm)^\Tsf(\xbm-\ybm) \geq \beta\|\Tsf\xbm-\Tsf\ybm\|^2, \quad \xbm, \ybm \in \R^n.$$
When $\beta = 1$, we say that $\Tsf$ is firmly nonexpansive.
\end{definition}

\medskip\noindent
The following results are derived from the definition above.

\medskip
\begin{proposition}
\label{Prop:BlockConvNonexp}
Let $\Tsf_i: \R^n \rightarrow \R^n$ for $i \in I$ be a set of nonexpansive operators. Then, their convex combination
$$\Tsf \defn \sum_{i \in I} \theta_i \Tsf_i, \quad\text{with}\quad \theta_i > 0 \text{ and } \sum_{i \in I} \theta_i = 1,$$
is nonexpansive.
\end{proposition}

\begin{proof}
By using the triangular inequality and the definition of nonexpansiveness, we obtain
\begin{align*}
&\|\Tsf\xbm-\Tsf\ybm\| \leq \sum_{i \in I} \theta_i \|\Tsf_i\xbm-\Tsf_i\ybm\|  \\
&\leq \left(\sum_{i \in I}\theta_i\right) \|\xbm-\ybm\| = \|\xbm-\ybm\|,
\end{align*}
for all $\xbm, \ybm \in \R^n$.
\end{proof}

\begin{proposition}
\label{Prop:NonexpCocoerOp}
Consider $\Rsf = \Isf - \Tsf$ where $\Tsf: \R^n \rightarrow \R^n$.
$$\Tsf \text{ is nonexpansive } \,\Leftrightarrow\, \Rsf \text{ is $(1/2)$-cocoercive.}$$
\end{proposition}

\begin{proof}
First suppose that $\Rsf$ is $1/2$ cocoercive. Let $\hbm \defn \xbm - \ybm$ for any $\xbm, \ybm \in \R^n$. We then have
$$\frac{1}{2}\|\Rsf\xbm-\Rsf\ybm\|^2 \leq (\Rsf\xbm-\Rsf\ybm)^\Tsf\hbm = \|\hbm\|^2 - (\Tsf\xbm-\Tsf\ybm)^\Tsf\hbm.$$
We also have that
$$\frac{1}{2}\|\Rsf\xbm-\Rsf\ybm\|^2 = \frac{1}{2}\|\hbm\|^2 - (\Tsf\xbm-\Tsf\ybm)^\Tsf\hbm + \frac{1}{2}\|\Tsf\xbm-\Tsf\ybm\|^2.$$
By combining these two and simplifying the expression
$$\|\Tsf\xbm-\Tsf\ybm\| \leq \|\hbm\|.$$
The converse can be proved by following this logic in reverse.
\end{proof}

\begin{definition}
For a constant $\alpha \in (0, 1)$, we say that $\Tsf$ is $\alpha$-averaged, if there exists a nonexpansive operator $\Nsf$ such that $\Tsf = (1-\alpha)\Isf + \alpha \Nsf$.
\end{definition}

The following characterization is often convenient.
\begin{proposition}
\label{Prop:BlockAveragedEquiv}
For a nonexpansive operator $\Tsf$, a constant $\alpha \in (0, 1)$, and the operator ${\Rsf \defn \Isf-\Tsf}$, the following are equivalent
\begin{enumerate}[label=(\alph*), leftmargin=*]
\item $\Tsf$ is $\alpha$-averaged
\item $(1-1/\alpha)\Isf + (1/\alpha)\Tsf$ is nonexpansive
\item $\|\Tsf\xbm - \Tsf\ybm\|^2 \leq \|\xbm-\ybm\|^2 - \left(\frac{1-\alpha}{\alpha}\right)\|\Rsf\xbm-\Rsf\ybm\|^2,\, \xbm, \ybm \in \R^n$.
\end{enumerate}
\end{proposition}

\begin{proof}
See Proposition~4.35 in~\cite{Bauschke.Combettes2017}.
\end{proof}

\begin{proposition}
\label{Prop:NonexpEquiv}
Consider $\Tsf: \R^n \rightarrow \R^n$ and $\beta > 0$. Then, the following are equivalent
\begin{enumerate}[label=(\alph*), leftmargin=*]
\item $\Tsf$ is $\beta$-cocoercive
\item $\beta\Tsf$ is firmly nonexpansive
\item $\Isf-\beta\Tsf$ is firmly nonexpansive.
\item $\beta\Tsf$ is $(1/2)$-averaged.
\item $\Isf-2\beta\Tsf$ is nonexpansive.
\end{enumerate}
\end{proposition}

\begin{proof}
For any $\xbm, \ybm \in \R^n$, let $\hbm \defn \xbm-\ybm$. The equivalence between (a) and (b) is readily observed by defining $\Psf \defn \beta\Tsf$ and noting that
\begin{align}
&(\Psf\xbm - \Psf\ybm)^\Tsf\hbm = \beta(\Tsf\xbm - \Tsf\ybm)^\Tsf\hbm \nonumber\\
&\quad\text{and}\quad \|\Psf\xbm-\Psf\ybm\|^2 = \beta^2 \|\Tsf\xbm-\Tsf\ybm\|.
\end{align}

\medskip\noindent
Define $\Rsf \defn \Isf - \Psf$ and suppose (b) is true, then
\begin{align*}
&(\Rsf\xbm-\Rsf\ybm)^\Tsf\hbm
= \|\hbm\|^2 - (\Psf\xbm-\Psf\ybm)^\Tsf\hbm \\
&= \|\Rsf\xbm-\Rsf\ybm\|^2 + (\Psf\xbm-\Psf\ybm)^\Tsf\hbm - \|\Psf\xbm-\Psf\ybm\|^2 \\
&\geq \|\Rsf\xbm-\Rsf\ybm\|^2.
\end{align*}
By repeating the same argument for $\Psf = \Isf - \Rsf$, we establish the full equivalence between (b) and (c).

\medskip\noindent
The equivalence of (b) and (d) can be seen by noting that
\begin{align*}
&2\|\Psf\xbm-\Psf\ybm\|^2 \leq 2(\Psf\xbm-\Psf\ybm)^\Tsf\hbm \\
&\Leftrightarrow\quad\|\Psf\xbm-\Psf\ybm\|^2 \leq 2(\Psf\xbm-\Psf\ybm)^\Tsf\hbm - \|\Psf\xbm-\Psf\ybm\|^2 \\
&= \|\hbm\|^2-(\|\hbm\|^2 - 2(\Psf\xbm-\Psf\ybm)^\Tsf\hbm + \|\Psf\xbm-\Psf\ybm\|^2)\nonumber\\
&= \|\hbm\|^2 - \|\Rsf\xbm-\Rsf\ybm\|^2.
\end{align*}

\medskip\noindent
To show the equivalence with (e), first suppose that ${\Nsf \defn \Isf - 2 \Psf}$ is  nonexpansive, then ${\Psf = \frac{1}{2}(\Isf + (-\Nsf))}$ is $1/2$-averaged, which means that it is firmly nonexpansive. On the other hand, if $\Psf$ is firmly nonexpansive, then it is $1/2$-averaged, which means that from Proposition~\ref{Prop:BlockAveragedEquiv}(b) we have that $(1-2)\Isf + 2\Psf = 2\Psf - \Isf = -\Nsf$ is nonexpansive. This directly means that $\Nsf$ is nonexpansive.
\end{proof}

\section{Proof of Theorem~\ref{Thm:ConvThm1}}

\label{Sec:Proof}
We consider the following operators
$$\Gsf \defn \nabla g + \Hsf \quad\text{and}\quad \Gsfhat \defn \nablahat g + \Hsf \quad\text{with}\quad \Hsf \defn \tau(\Isf - \Dsf_\sigma),$$
where $\Gsfhat$ is the minibatch approximation of $\Gsf$. The direct application of Assumption~\ref{As:DataFitConvexity} implies that for any $\xbm, \ybm \in \R^n$
\begin{subequations}
\label{Eq:Expectations}
\begin{align}
&\E[\Gsfhat (\xbm)] = \E[\nablahat g(\xbm)] + \Hsf(\xbm) = \Gsf(\xbm)\\
&\E[\|\Gsf(\xbm) - \Gsfhat (\xbm)\|_2^2] = \E[\|\nabla g(\xbm) - \nablahat g(\xbm)\|_2^2] \leq \frac{\nu^2}{B}
\end{align}
\end{subequations}

Now we prove Theorem~\ref{Thm:ConvThm1} in several steps.
\begin{enumerate}[label=(\alph*), leftmargin=*]
\item Since $\nabla g$ is $L$-Lipschitz continuous, we know that it is $(1/L)$-cocoercive (see Theorem~2.1.5 in Section~2.1 of~\cite{Nesterov2004}). Then from Proposition~\ref{Prop:NonexpEquiv}, we know that the operator ${(\Isf-(2/L)\nabla g)}$ is nonexpansive.
\item From the definition of $\Hsf$ and the fact that $\Dsf$ is nonexpansive, we know that ${(\Isf-(1/\tau)\Hsf) = \Dsf}$ is nonexpansive.

\item From Proposition~\ref{Prop:BlockConvNonexp}, we know that a convex combination of nonexpansive operators is also nonexpansive, hence
\begin{align*}
&\Isf - \frac{2}{L+2\tau}\Gsf = \left(\frac{2}{L+2\tau}\cdot \frac{L}{2}\right)\left[\Isf-\frac{2}{L}\nabla g\right] \\
&+ \left(\frac{2}{L+2\tau}\cdot \frac{2\tau}{2}\right)\left[\Isf - \frac{1}{\tau}\Hsf\right]\nonumber,
\end{align*}
is nonexpansive. Then from Proposition~\ref{Prop:NonexpEquiv}, we know that $\Gsf$ is $1/(L+2\tau)$-cocoercive.

\item Consider any $\xbmast \in \zer(\Gsf)$ and $\xbm \in \R^n$. We then have
\begin{align}
\label{Eq:SingleIter}
\nonumber&\|\xbm - \xbmast - \gamma \Gsf\xbm\|^2\\
\nonumber&= \|\xbm-\xbmast\|^2 - 2\gamma (\Gsf\xbm-\Gsf\xbmast)^\Tsf(\xbm-\xbmast) + \gamma^2\|\Gsf\xbm\|^2 \\
\nonumber&\leq \|\xbm-\xbmast\|^2 - \frac{2\gamma-(L+2\tau)\gamma^2}{L+2\tau}\|\Gsf\xbm\|^2 \\
&\leq \|\xbm-\xbmast\|^2 - \frac{\gamma}{L+2\tau}\|\Gsf \xbm\|^2,
\end{align}
where we used $\Gsf\xbmast = \zerobm$, the cocoercivity of $\Gsf$, and the fact that $0 < \gamma \leq 1/(L+2\tau)$. 

\item For a single iteration of SIMBA $\xbm^+ = \xbm - \gamma\Gsfhat\xbm$, we have
\begin{align*}
&\|\xbm^+-\xbmast\|^2 = \|\xbm - \xbmast - \gamma \Gsfhat\xbm\|^2\\
&=\|\xbm - \xbmast - \gamma \Gsf\xbm + \gamma(\Gsf\xbm-\Gsfhat\xbm)\|^2\\
&= \|\xbm - \xbmast - \gamma \Gsf\xbm\|^2 + \gamma^2 \|\Gsf\xbm-\Gsfhat\xbm\|^2\\ 
&+ 2\gamma(\Gsf\xbm-\Gsfhat\xbm)^\Tsf(\xbm - \xbmast - \gamma \Gsf\xbm).
\end{align*}
By taking the conditional expectation with respect to the previous iterate $\xbm$, using~\eqref{Eq:Expectations}, and applying the bound~\eqref{Eq:SingleIter}, we obtain
\begin{align}
\label{Eq:Expectation}
\nonumber&\E\left[\|\xbm^+-\xbmast\|^2 \mid \xbm\right] \leq \|\xbm - \xbmast - \gamma \Gsf\xbm\|^2 + \frac{\gamma^2\nu^2}{B} \\
&\leq \|\xbm-\xbmast\|^2 - \frac{\gamma}{L+2\tau}\|\Gsf \xbm\|^2 + \frac{\gamma^2\nu^2}{B}.
\end{align}

\item By simply rearranging the terms, we obtain
\begin{align*}
&\frac{\gamma}{L + 2\tau} \|\Gsf\xbm\|^2 \leq   \E\left[\|\xbm-\xbmast\|^2 - \|\xbm^+ - \xbmast\|^2 | \xbm \right] + \frac{\gamma^2\nu^2}{B}
\end{align*}
\item Hence, by averaging over $t \geq 1$ iterations and taking the total expectation, we obtain our main result
\begin{align*}
\E\left[\frac{1}{t}\sum_{k = 1}^{t} \|\Gsf\xbm^{k-1}\|^2\right] 
\leq \frac{(L+2\tau)}{\gamma} \left[\frac{\|\xbm^0-\xbmast\|^2}{t} + \frac{\gamma^2\nu^2}{B}\right].
\end{align*}
\end{enumerate}

\end{document}